\documentclass{llncs}
\pdfoutput=1 
\usepackage{microtype}
\usepackage{plaatjes}
\usepackage{amsmath}
\usepackage{algorithmic}
\usepackage{color}
\usepackage{hyperref}
\newcommand{\abs}[1]{{\left\vert #1 \right\vert}}
\usepackage{mathptmx} 

\let\oldendproof\endproof
\def\endproof{\qed\oldendproof}

\newtheorem {obs} {Observation}

\newcommand{\personalremark}[3]{}

\newcommand{\david}[2][says]{\personalremark{David}{#1}{#2}}
\newcommand{\maarten}[2][says]{\personalremark{Maarten}{#1}{#2}}

\title{Listing All Maximal Cliques in Sparse Graphs in Near-optimal Time}
\author{David Eppstein, Maarten L\"offler, and Darren Strash}
\institute{Department of Computer Science, University of California, Irvine, USA}
\begin{document}
\maketitle

\begin{abstract}
The \emph{degeneracy} of an $n$-vertex graph $G$ is the smallest number $d$ such that every subgraph of $G$ contains a vertex of degree at most $d$.
We show that there exists a nearly-optimal fixed-parameter
tractable algorithm for enumerating all maximal cliques, parametrized by degeneracy. To achieve this result,
we modify the classic Bron--Kerbosch algorithm and show that it
runs in time $O(dn3^{d/3})$. We also provide matching upper and lower bounds showing that the largest possible number of maximal cliques in an $n$-vertex graph with degeneracy~$d$ (when $d$ is a multiple of 3 and $n\ge d+3$) is $(n-d)3^{d/3}$. Therefore, our algorithm matches the $\Theta(d(n-d)3^{d/3})$
worst-case output size of the problem whenever $n-d=\Omega(n)$.
\end{abstract}

\section{Introduction}
Cliques, complete subgraphs of a graph, are of great importance in many applications. In social networks
cliques may represent closely connected clusters of actors~\cite{BerKoMoy-WAO-04,DuWuPei-WMSNA-07,harary57,luce49} and may be used as features in exponential random graph models
for statistical analysis of social networks~\cite{eppstein2009,frank91,frank86,robins2007,wasserman96}. In bioinformatics, clique finding procedures have been used to detect structural motifs from protein similarities~\cite{grindley93,koch2001,KocLenWan-JCB-96}, to predict unknown protein structures~\cite{samudrala98}, and to determine the docking regions where two biomolecules may connect to each other~\cite{gardiner99}. Clique finding problems also arise in document clustering~\cite{auguston70}, in the recovery of depth from stereoscopic image data~\cite{HorSko-PAMI-89}, in computational topology~\cite{Zom-SoCG-10}, and in e-commerce, in the discovery of patterns of items that are frequently purchased together~\cite{ZakParOgi-KDD-97}.

Often, it is important to find not just one large clique, but all \emph{maximal cliques}.
Many algorithms are now known for this problem~\cite{akkoyunlu73,bron73,cazals2008,chiba85,chrobak91,gerhards79,harary57,johnston75,makino2004,mulligan72,tomita2006} and for the complementary problem of finding maximal independent sets~\cite{Epp-TALG-09,johnson88,lawler80,loukakis81,tsukiyama77}. One of the most successful in practice is the \emph{Bron--Kerbosch algorithm}, a
simple backtracking
procedure that recursively solves subproblems specified by three sets of vertices: the vertices that are required to be included in a partial clique, the vertices that are to be excluded from the clique, and some remaining vertices whose status still needs to be determined~\cite{bron73,cazals2008,johnston75,koch2001,tomita2006}.

All maximal cliques can be listed in polynomial time per clique~\cite{lawler80,tsukiyama77} or in a total time proportional to the maximum possible number of cliques in an $n$-vertex graph, without additional polynomial factors~\cite{Epp-JGAA-03,tomita2006}. In particular, a variant of the Bron--Kerbosch algorithm is known to be optimal in this sense~\cite{cazals2008,tomita2006}. Unfortunately this maximum possible number of cliques is exponential~\cite{moon-moser65}, so that all general-purpose algorithms for listing maximal cliques necessarily take exponential time.

We are faced with a dichotomy between theory, which states categorically that clique finding takes exponential time, and practice, according to which clique finding is useful and can be efficient in its areas of application. One standard way of resolving dilemmas such as this one is to apply \emph{parametrized complexity}~\cite{DowFel-99}: one seeks a parameter of instance complexity such that instances with small parameter values can be solved quickly. A parametrized problem is said to be \emph{fixed-parameter tractable} if instances with size $n$ and parameter value $p$ can be solved in a time bound of the form $f(p) n^{O(1)}$, where $f$ may grow exponentially or worse with~$p$ but is independent of~$n$. With this style of analysis, instances with a small parameter value are used to model problems that can be solved quickly, while instances with a large parameter value represent a theoretical worst case that, one hopes, does not arise in practice.

The size of the largest clique does not work well as a parameter: the maximum clique problem, parametrized by clique size, is hard for
${\mathbf W}[1]$, implying that it is unlikely to have a fixed-parameter
tractable algorithm~\cite{downey95}, and Tur\'an graphs\ $K_{\frac{n}{k},\frac{n}{k},\frac{n}{k},\dots}$ have $(n/k)^k$ maximal cliques of size $k$ forcing any algorithm that lists them all to take time larger than any fixed-parameter-tractable bound. However, clique size is not the only parameter one can choose. In this paper, we study maximal clique finding parametrized by \emph{degeneracy}, a frequently-used measure of the sparseness of a graph that is closely related to other common sparsity measures such as arboricity and thickness, and that has previously been used for other fixed-parameter problems~\cite{AloGut-Algo-09,CaiChaCha-IWPEC-06,GolVil-WG-08,KloCai-ICS-00}. We are motivated by the fact that sparse graphs often appear in practice. For instance, the World Wide Web graph, citation networks, and collaboration graphs have low arboricity~\cite{goel2006}, and therefore have low degeneracy. Empirical evidence also suggests that the $h$-index, a measure of sparsity that upper bounds degeneracy, is low for social networks~\cite{eppstein2009}. As we show in Appendix~\ref{sec:appendix}, protein--protein interaction networks have low degeneracy as well. Furthermore, planar graphs have degeneracy at most five~\cite{LicWhi-CJM-70}, and the Barab\'asi--Albert model of preferential attachment~\cite{barabasi99}, frequently used as a model for large scale-free social networks, produces graphs with bounded degeneracy.  
 We show that:
\begin{itemize}
\item A variant of the Bron--Kerbosch algorithm, when applied to $n$-vertex graphs with degeneracy $d$, lists all maximal cliques in time $O(dn3^{d/3})$.
\item Every $n$-vertex graph with degeneracy $d$ (where $d$ is a multiple of three and $n\ge d+3$) has at most $(n-d)3^{d/3}$ maximal cliques, and there exists an $n$-vertex graph with degeneracy $d$ that has exactly $(n-d)3^{d/3}$ maximal cliques. Therefore, our variant of the Bron--Kerbosch algorithm is optimal in the sense that its time is within a constant of the parametrized worst-case output size.
\end{itemize}
Our algorithms are fixed-parameter tractable, with a running time of the form $O(f(d)n)$ where $f(d)=d3^{d/3}$. Algorithms for listing all maximal cliques in graphs of constant degeneracy in time $O(n)$ were already known~\cite{chiba85,chrobak91}, but these algorithms had not been analyzed for their dependence on the degeneracy of the graph. We compare the parametrized running time bounds of the known alternative algorithms to the running time of our variant of the Bron--Kerbosch algorithm, and we show that the Bron--Kerbosch algorithm has a much smaller dependence on the parameter~$d$. Thus we give theoretical evidence for the good performance for this algorithm that had previously been demonstrated empirically.

\section{Preliminaries}
We work with an undirected graph $G=(V,E)$, which we assume
is stored in an adjacency list data structure. We let $n$ and $m$
be the number of vertices and edges of $G$, respectively. For a
vertex $v$, we define $\Gamma(v)$ to be the set $\{w\mid (v,w)\in E\}$, which
we call the \emph{neighborhood} of $v$, and similarly for a subset $W \subset V$ we define $\Gamma(W)$ to be the set $\bigcap_{w \in W} \Gamma(w)$, which is the common neighborhood of all vertices in $W$.

\subsection{Degeneracy}
Our algorithm is parametrized on the \emph{degeneracy} of a graph,
a measure of its sparsity.

\begin{definition}[degeneracy]
The degeneracy of a graph $G$ is the smallest value $d$ such that every nonempty subgraph of $G$ contains a vertex of degree at most $d$~\cite{LicWhi-CJM-70}.
\end{definition}

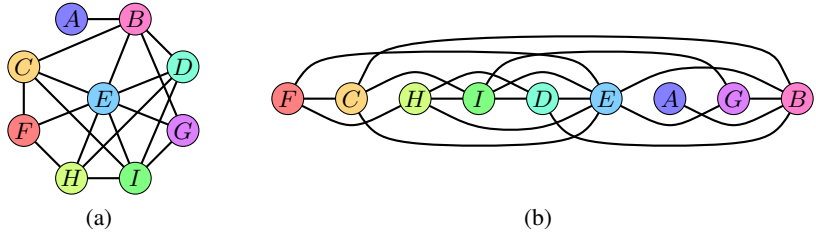
\tweeplaatjes {graph-classic} {graph-ordered} {(a) A graph with degeneracy $3$.
(b) A vertex ordering showing that the degeneracy is not larger than $3$.}

Figure~\ref {fig:graph-classic} shows an example of a graph of degeneracy $3$.
Degeneracy is also known as the $k$-core number~\cite{batagelj2003}, width~\cite{freuder82}, and linkage~\cite{kirousis96} of a graph 
and is one less than the coloring number~\cite{erdos66}.
In a graph of degeneracy $d$, the maximum clique size can be at most $d+1$, for any larger clique would form a subgraph in which all vertices have degree higher than $d$.

If a graph has degeneracy $d$, then it has a \emph{degeneracy ordering}, an ordering such that each vertex has $d$ or fewer neighbors that come later in the ordering.
Figure~\ref {fig:graph-ordered} shows a possible degeneracy ordering for the example.
Such an ordering may be formed from $G$ by repeatedly removing a vertex of degree $d$ or less: by the assumption that $G$ is $d$-degenerate, at least one such vertex exists at each step. Conversely, if $G$ has an ordering with this property, then it is $d$-degenerate, because for any subgraph $H$ of $G$, the vertex of $H$ that comes first in the ordering has $d$ or fewer neighbors in $H$. Thus, as Lick and White~\cite{LicWhi-CJM-70} showed, the degeneracy may equivalently be defined as the minimum $d$ for which a degeneracy ordering exists. A third, equivalent definition is that $d$ is the minimum value for which there exists an orientation of $G$ as a directed acyclic graph in which all vertices have out-degree at most $d$~\cite{chrobak91}: such an orientation may be found by orienting each edge from its earlier endpoint to its later endpoint in a degeneracy ordering, and conversely if such an orientation is given then a degeneracy ordering may be found as a topological ordering of the oriented graph.

Degeneracy is a robust measure of sparsity: it is
within a constant factor of other popular measures of sparsity
including arboricity and thickness.
In addition, degeneracy, along with a degeneracy ordering,
can be computed by a simple greedy strategy of repeatedly removing a vertex
with smallest degree (and its incident edges) from the graph until it is
empty. The degeneracy is the maximum of the degrees of the vertices at the time they
are removed from the graph, and the degeneracy ordering is the order
in which vertices are removed from the graph~\cite{jensen95}. The easy computation of degeneracy has made it a useful tool in algorithm design and analysis~\cite{chrobak91,Epp-TALG-09}.

We can implement this algorithm in $O(n+m)$ time by maintaining an array $D$,
where $D[i]$ stores a list of vertices of degree $i$ in the graph~\cite{batagelj2003}. To remove a vertex 
of minimum degree from the graph, we scan from the beginning of the array until we 
reach the first nonempty list, remove a vertex from this list, and then update its neighbors' 
degrees and move them to the correct lists. Each vertex removal step takes time proportional to the degree
of the removed vertex, and therefore the algorithm takes linear time.

By counting the maximum possible number of edges from each vertex to later neighbors, we get the following bound on the number of edges of a $d$-degenerate graph:

\begin{lemma}[Proposition 3 of \cite{LicWhi-CJM-70}] A graph $G=(V,E)$ with degeneracy $d$ has at most
$d(n-\frac{d+1}{2})$ edges.
\end{lemma}

\subsection{The Bron--Kerbosch Algorithm}
The Bron--Kerbosch algorithm~\cite{bron73} is a widely used algorithm
for finding all maximal cliques in a graph. It is a
recursive backtracking algorithm which is easy to understand,
easy to code, and has been shown to work well in practice.

A recursive call to the Bron--Kerbosch algorithm provides three disjoint sets of vertices $R$, $P$, and $X$ as arguments, where $R$ is a (possibly non-maximal) clique and $P\cup X = \Gamma(R)$ are the vertices that are adjacent to every vertex in $R$. The vertices in $P$ will be considered to be added to clique $R$, while those in $X$ must be excluded from the clique; thus, within the recursive call, the algorithm lists all cliques in $P\cup R$ that are maximal within the subgraph induced by $P\cup R\cup X$.
The algorithm chooses a candidate $v$ in $P$ to add to the clique $R$, and
makes a recursive call in which $v$ has been moved from $R$ to $P$; in this recursive call, it restricts $X$ to the neighbors of $v$, since non-neighbors cannot affect the maximality of the resulting cliques. When the recursive call returns, $v$ is moved to $X$ to eliminate redundant work by further calls
to the algorithm. When the recursion reaches a level at which $P$ and $X$ are empty, $R$ is a
maximal clique and is reported (see Fig.~\ref{figure:bkalg}). To list all maximal cliques in the graph, this recursive algorithm is called with $P$ equal to the set of all vertices in the graph and with $R$ and $X$ empty.

\begin{figure}
{\bf proc} BronKerbosch($P$, $R$, $X$)
\begin{algorithmic}[1]
\IF{$P\cup X = \emptyset$}
    \STATE report $R$ as a maximal clique
\ENDIF
\FOR{ {\bf each} vertex $v\in P$}
    \STATE BronKerbosch($P\cap \Gamma(v)$, $R\cup\{v\}$, $X\cap \Gamma(v)$)
    \STATE $P \leftarrow P \setminus \{v\}$
    \STATE $X \leftarrow X \cup \{v\}$
\ENDFOR
\end{algorithmic}

\bigskip

{\bf proc} BronKerboschPivot($P$, $R$, $X$)
\begin{algorithmic}[1]
\IF{$P\cup X = \emptyset$}
    \STATE report $R$ as a maximal clique
\ENDIF
\STATE choose a pivot $u \in P\cup X$  \COMMENT{Tomita et al. choose $u$ to maximize $\abs{P\cap\Gamma(u)}$}
\FOR{ {\bf each} vertex $v\in P\setminus \Gamma(u)$}
    \STATE BronKerboschPivot($P\cap \Gamma(v)$, $R\cup\{v\}$, $X\cap \Gamma(v)$)
    \STATE $P \leftarrow P \setminus \{v\}$
    \STATE $X \leftarrow X \cup \{v\}$
\ENDFOR
\end{algorithmic}

\caption{The Bron--Kerbosch algorithm without and with pivoting}
\label{figure:bkalg}
\end{figure}
Bron and Kerbosch also describe a heuristic called
\emph{pivoting}, which limits the number of recursive
calls made by their algorithm. The key observation is that
for any vertex $u$ in $P\cup X$, called a \emph{pivot}, any maximal clique must
contain one of $u$'s non-neighbors (counting $u$ itself as a non-neighbor). Therefore,
we delay the vertices in $P\cap\Gamma(u)$ from being
added to the clique until future recursive calls, with the benefit
that we make fewer recursive calls.
Tomita et al.~\cite{tomita2006} show that choosing the pivot $u$
from $P\cup X$ in order to maximize $\abs{P \cap \Gamma(u)}$ guarantees that
the Bron--Kerbosch algorithm has worst-case running time $O(3^{n/3})$,
excluding time to write the output, which is worst-case optimal.

\section{The Algorithm}

In this section, we show that apart from the pivoting strategy, the order in which the vertices of $G$ are processed by the Bron--Kerbosch algorithm is also important. By choosing an ordering carefully, we develop a variant of the Bron--Kerbosch algorithm that correctly lists all maximal cliques in time $O(dn3^{d/3})$. Essentially, our algorithm performs the outer level of recursion of the Bron--Kerbosch algorithm without pivoting, using a degeneracy ordering to order the sequence of recursive calls made at this level, and then switches at inner levels of recursion to the pivoting rule of Tomita et al.~\cite{tomita2006}.


In the original Bron--Kerbosch algorithm, in each recursive call the vertices in $P$ are considered for expansion one by one (see line 4 of BronKerbosch in Figure~\ref {figure:bkalg}). The order in which the vertices are treated is not specified. We first analyze what happens if we fix an order $v_1, v_2, \ldots, v_n$ on the vertices of $V$, and use the same order consistently to loop through the vertices of $P$ in each recursive call of the non-pivoting version of BronKerbosch.

\begin {obs} \label {obs:XP-order}
When processing a clique $R$ in the ordered variant of Bron--Kerbosch, the common neighbors of $R$ can be partitioned into the set $P$ of vertices that come after the last vertex of $R$, and the set $X$ of remaining neighbors, as shown in Figure~\ref{fig:BK-ordered}.    \end {obs}

\maarten [wonders] {Should we prove this?}

\david{I think it's ok to leave it as-is. It would seem silly to make an appendix just to prove something that is (in retrospect) so obvious.}

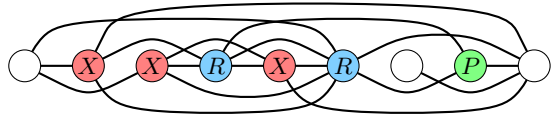
\eenplaatje {BK-ordered} {Partitioning the common neighbors of a clique $R$ into the set $P$ of later vertices and the set $X$ of remaining neighbors.}

%
%
%
Our algorithm computes a degeneracy ordering of the given graph, and performs the outermost recursive calls in the ordered variant of the Bron--Kerbosch algorithm (without pivoting) for this ordering.
The sets $P$ passed to each of these recursive calls will have at most $d$ elements in them, leading to few recursive calls within each of these outer calls. Below the top level of the recursion we switch from the ordered non-pivoting version of the Bron--Kerbosch algorithm to the pivoting algorithm (with the same choice of pivots as Tomita et al.~\cite{tomita2006}) to further control the number of recursive calls.

\begin{figure}[tbh!]
{\bf proc} BronKerboschDegeneracy($V$, $E$)
\begin{algorithmic}[1]
\FOR{each vertex $v_i$ in a degeneracy ordering $v_0$, $v_1$, $v_2$, \dots of $(V,E)$}
    \STATE{$P\leftarrow \Gamma(v_i)\cap \{v_{i+1},\ldots, v_{n-1}\}$ }
    \STATE{$X\leftarrow \Gamma(v_i)\cap \{v_0,\ldots, v_{i-1}\}$ }
    \STATE{BronKerboschPivot($P$, $\{v_i\}$, $X$)}
\ENDFOR
\end{algorithmic}

\caption{Our algorithm.}
\end{figure}



\begin{lemma}
\label{lemma:correctBK}
The Bron--Kerbosch algorithm using the Tomita et al. pivoting strategy generates
all and only maximal cliques containing all vertices in $R$, some vertices in $P$,
and no vertices in $X$, without duplication.
\end{lemma}
\begin{proof} See Tomita et al.~\cite{tomita2006}.
\end{proof}

\begin{theorem} Algorithm \emph{BronKerboschDegeneracy}
generates all and only maximal cliques without duplication.
\end{theorem}

\begin{proof}
Let $C$ be a maximal clique, and $v$ its earliest vertex in the degeneracy order.
By Lemma~\ref{lemma:correctBK}, $C$ will be reported (once) when processing $v$.
When processing any other vertex of $C$, $v$ will be in $X$, so $C$ will not be reported.
\end{proof}

To make pivot selection fast we pass as an additional argument to BronKerboschPivot a subgraph $H_{P,X}$ of $G$ that has $P\cup X$ as its vertices; an edge $(u,v)$ of $G$ is kept as an edge in $H_{P,X}$ whenever at least one of $u$ or $v$ belongs to $P$ and both of them belong to $P\cup X$. The pivot chosen according to the pivot rule of Tomita et al.~\cite{tomita2006} is then just the vertex in this graph with the most neighbors in $P$.

\begin{lemma}
\label{lem:hpx1}
Whenever \emph{BronKerboschDegeneracy} calls \emph{BronKerboschPivot} it can form $H_{P,X}$ in time $O(d(\abs{P} + \abs{X}))$.
\end{lemma}

\begin{proof}
The vertex set of $H_{P,X}$ is known from $P$ and $X$.
Each edge is among the $d$ outgoing edges from each of its vertices.
Therefore, we can achieve the stated time bound by looping over each of the $d$ outgoing edges from each vertex in $P\cup X$ and testing whether each edge meets the criterion for inclusion in $H_{P,X}$.
\end{proof}

The factor of $d$ in the time bound of Lemma~\ref{lem:hpx1} makes it too slow for the recursive part of our algorithm. Instead we show that the subgraph to be passed to each recursive call can be computed quickly from the subgraph given to its parent in the recursion tree.

\begin{lemma}
\label{lem:hpx2}
In a recursive call to \emph{BronKerboschPivot} that is passed the graph $H_{P,X}$ as an auxiliary argument, the sequence of graphs $H_{P\cap\Gamma(v),X\cap\Gamma(v)}$ to be passed to lower-level recursive calls can be computed in total time $O(|P|^2(|P|+|X|))$.
\end{lemma}

\begin{proof}
It takes $O(|P|+|X|)$ time to identify the subsets $P\cap\Gamma(v)$ and $X\cap\Gamma(v)$ by examining the neighbors of $v$ in $H_{P,X}$.
Once these sets are identified, $H_{P\cap\Gamma(v),X\cap\Gamma(v)}$ may be constructed as a subgraph of $H_{P,X}$ in time $O(|P|(|P|+|X|))$ by testing for each edge of $H_{P,X}$ whether its endpoints belong to these sets. There are $O(|P|)$ graphs to construct, one for each recursive call, hence the total time bound.
\end{proof}

\begin{lemma}[Theorem 3 of~\cite{tomita2006}]
\label{lemma-tomita-runtime}
Let $T$ be a function which satisfies the following recurrence relation:
\[ T(p) \leq \begin{cases}\max_{k}\{kT(p-k)\} + dp^2 &\text{if $p>0$}\\
         e &\text{if $p=0$}\end{cases}\]
Where $p$ and $k$ are integers, such that $p\geq k$, and $d,e$ are constants greater than zero.
Then, $T(p) \leq \max_{k}\{kT(p-k)\} + dp^2 = O(3^{p/3})$.
\end{lemma}

\begin{lemma}
\label{lemma-runtime}
Let $v$ be a vertex, $P_v$, be $v$'s later neighbors, and $X_v$ be
$v$'s earlier neighbors. Then \emph{BronKerboschPivot($P_v$, $\{v\}$, $X_{v}$)} executes
in time $O((d + \abs{X_v})3^{\abs{P_v}/3})$, excluding the time to report the discovered
maximal cliques.
\end{lemma}

\begin{proof}
Define $D(p,x)$ to be the running time of BronKerboschPivot($P_v$, $\{v\}$, $X_v$),
where $p = \abs{P_v}$, and $x=\abs{X_v}$. We show that $D(p,x) =O((d+x)3^{p/3})$.
By the description of BronKerboschPivot, $D$ satisfies the following recurrence relation:

\[ D(p,x) \leq \begin{cases}\max_{k}\{kD(p-k,x)\} + c_1p^2(p+x) &\text{if $p>0$}\\
         c_2 &\text{if $p=0$}\end{cases}\]
where $c_1$ and $c_2$ are constants greater than 0.

Since our graph has degeneracy $d$, the inequality $p+x\leq d+x$ always holds. Thus,
\begin{align*}D(p,x) &\leq \max_{k}\{kD(p-k,x)\} + c_1p^2(p+x) \\
                     &\leq (d+x)\left( \max_{k}\left\{\frac{kD(p-k,x)}{d+x}\right\} + c_1p^2\right) \\
                     &\leq (d+x)\left(\max_{k}\{kT(p-k)\} + c_1p^2\right) \\
                     &=O((d+x)3^{p/3})\text{\quad by letting $d=c_1,e=c_2$ in Lemma~\ref{lemma-tomita-runtime}}
\end{align*}
\end{proof}

%
%
%

\begin{theorem}
\label{theorem-totaltime}
Given a $n$-vertex graph $G$ with degeneracy $d$, our algorithm
reports all maximal cliques of $G$ in time $O(dn3^{d/3})$.
\end{theorem}

\begin{proof}
For each initial call to BronKerboschPivot for each
vertex $v$, we first spend time $O(d(\abs{P_v}+\abs{X_v}))$ to
set up subgraph $H_{P_v,X_v}$. Over the entire algorithm
we spend time
\[\sum_vO(d(\abs{P_v}+\abs{X_v})) = O(dm) = O(d^2n)\]
setting up these subgraphs. The time spent performing the recursive calls is
\[\sum_vO((d+\abs{X_v})3^{\abs{P_v}/3}) = O((dn+m)3^{d/3}) = O(dn3^{d/3}),\]
and the time to report all cliques is $O(d\mu)$, where $\mu$ is the number
of maximal cliques. We show in the next section that $\mu=(n-d)3^{d/3}$
in the worst case, and therefore we take time $O(d(n-d)3^{d/3})$ reporting
cliques in the worst case. Therefore, the algorithm executes in time $O(dn3^{d/3})$.
\end{proof}

This running time is nearly worst-case optimal, since there may be
$\Theta((n-d)3^{d/3})$ maximal cliques in the worst case.

\section{Worst-case Bounds on the Number of Maximal Cliques}
\label{sec:num-cliques}
\begin{theorem}
\label{thm:num-cliques}
Let $d$ be a multiple of 3 and $n\ge d+3$.
Then the largest possible number of maximal cliques in an $n$-vertex graph with degeneracy $d$ is $(n-d)3^{d/3}$.
\end{theorem}

\begin{proof} We first show that there cannot be more than $(n-d)3^{d/3}$
maximal cliques. We then show that there exists a graph that has $(n-d)3^{d/3}$ maximal cliques.

\paragraph{An Upper Bound.}
Consider a degeneracy ordering
of the vertices, in which each vertex has at most $d$ neighbors
that come later in the ordering. For each vertex~$v$ that is placed among the first $n-d-3$ vertices of the degeneracy ordering, we count the number of maximal
cliques such that $v$ is the clique vertex that comes first in the ordering.

Since vertex $v$ has at
most $d$ later neighbors, the Moon--Moser bound~\cite{moon-moser65} applied to the subgraph induced by these later neighbors shows that they can form at most $3^{d/3}$ maximal cliques with each
other. Since these vertices are all neighbors of $v$, $v$
participates in at most $3^{d/3}$ maximal cliques with its later neighbors. Thus,
the first $n-d-3$ vertices contribute to at most $(n-d-3)3^{d/3}$ maximal
cliques total.

By the Moon--Moser bound, the remaining $d+3$ vertices
in the ordering can form at most $3^{(d+3)/3}$ maximal cliques.
Therefore, a graph with degeneracy $d$ can have no more than
$(n-d-3)3^{d/3} + 3^{(d+3)/3} = (n-d)3^{d/3}$ maximal cliques.

\paragraph{A Lower Bound.}
By a simple counting argument, we can see that the graph $K_{n-d,3,3,3,3,\ldots}$
contains $(n-d)3^{d/3}$ maximal cliques: Each maximal clique must contain exactly one vertex from
each disjoint independent set of vertices, and there are $(n-d)3^{d/3}$ ways
of forming a maximal clique by choosing one vertex from each independent set.
Figure~\ref {fig:lower-bound} shows this construction for $d = 6$.
We can also see that this graph is $d$-degenerate since in any ordering of the vertices,
the first vertex must have $d$ or more later neighbors, and in any ordering
where the $n-d$ disjoint vertices come first, these first $n-d$ vertices
have exactly $d$ later neighbors, and the last $d$ vertices have fewer later
neighbors.
\end{proof}

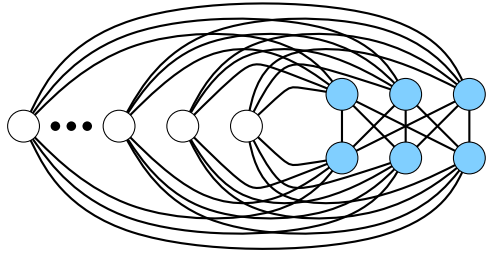
\eenplaatje {lower-bound} {The lower bound construction for $d = 6$, consisting of a Moon--Moser graph of size $d$ on the right (blue vertices) and an independent set of $n - d$ remaining vertices that are each connected to all of the last $d$ vertices.}

Relatedly, a bound of $(n-d+1)2^d$ on the number of cliques (without assumption of maximality) in $n$-vertex $d$-degenerate graphs was already known~\cite{Woo-GC-07}.

\section{Comparison with Other Algorithms}
\label{section:compare}

Chiba and Nishizeki~\cite{chiba85} describe two algorithms for finding cliques in sparse graphs. The first of these two algorithms reports all maximal
cliques using $O(am)$ time per clique, where $a$ is the arboricity
of the graph, and $m$ is the number of edges in $G$. The \emph{arboricity}
is the minimum number of edge-disjoint spanning forests into which the
graph can be decomposed~\cite{harary72}. The degeneracy of a graph is
closely related to arboricity: $a\le d\le 2a-1$.
In terms of degeneracy, Chiba and Nishizeki's algorithm
uses $O(d^2n)$ time per clique. Combining this with the bound on
the number of cliques derived in Section~\ref{sec:num-cliques} results in a worst-case time bound of $O(d^2n(n-d)3^{d/3})$. For constant $d$, this is a quadratic time bound, in contrast to the linear time of our algorithm.

Another algorithm of Chiba and Nishizeki~\cite{chiba85} lists cliques of order $l$ in time $O(la^{l-2}m)$. It can be adapted to enumerate all maximal cliques in a graph
with degeneracy $d$ by first enumerating all cliques of order
$d+1$, $d$, $\dots$ down to $1$, and removing cliques that are not maximal.
Applying their algorithm directly to a $d$-degenerate graph takes time $O(ld^{l-1}n)$.
Therefore, the running time to find all maximal cliques is
$\sum_{1\leq i\leq d+1} O(ind^{i-1}) = O(nd^{d+1}).$
Like our algorithm, this is linear when $d$ is constant, but with a much worse dependence on the parameter~$d$.

Chrobak and Eppstein~\cite{chrobak91}
list triangles and 4-cliques in graphs of bounded degeneracy
by testing all sets of two or three later neighbors of each vertex according to a degeneracy ordering.
The same idea extends in an obvious way to finding maximal cliques of size greater than four, by testing all subsets of later neighbors of each vertex. For each vertex $v$, there are at most $2^d$ subsets to test; each subset may be tested for being a clique in time $O(d^2)$, by checking whether each of its vertices has all the later vertices in the subset among its later neighbors, giving a total time of $O(nd^2 2^d)$ to list all the cliques in the graph. However, although this singly-exponential time bound is considerably faster than Chiba and Nishizeki, and is close to known bounds on the number of (possibly non-maximal) cliques in $d$-degenerate graphs~\cite{Woo-GC-07}, it is slower than our algorithm by a factor that is exponential in~$d$. Our new algorithm uses this same idea of searching among the later neighbors in a degeneracy order but achieves much greater efficiency by combining it with the Bron--Kerbosch algorithm.


\section{Conclusion}

We have presented theoretical evidence for the fast performance of the Bron--Kerbosch algorithm for finding cliques in graphs, as has been observed in practice. We observe that the problem is fixed-parameter tractable in terms of the degeneracy of the graph, a parameter that is expected to be low in many real-world applications, and that a slight modification of the Bron--Kerbosch algorithm performs optimally in terms of the degeneracy.

We explicitly prescribe the order in which the Bron--Kerbosch algorithm processes the vertices of the graph, something that has not been considered before. Without this particular order, we do not have a bound on the running time. It would be interesting to determine whether a random order gives similar results, as this would further explain the observed performance of implementations of Bron--Kerbosch that do not use the degeneracy order.


\section*{Acknowledgments}
This research was supported in part by the National Science
Foundation under grant 0830403, and by the
Office of Naval Research under MURI grant N00014-08-1-1015.

{\raggedright
\bibliographystyle{abuser}
\bibliography{bron-kerbosch}

\begin{thebibliography}{10}

\bibitem{akkoyunlu73}
E.~A. Akkoyunlu.
\newblock {The enumeration of maximal cliques of large graphs}.
\newblock {\em SIAM J. Comput.} 2(1):1{--}6, 1973,
  \href{http://dx.doi.org/10.1137/0202001}%
{doi:10.1137/0202001}.

\bibitem{AloGut-Algo-09}
N.~Alon and S.~Gutner.
\newblock {Linear time algorithms for finding a dominating set of fixed size in
  degenerated graphs}.
\newblock {\em Algorithmica} 54(4):544{--}556, 2009,
  \href{http://dx.doi.org/10.1007/s00453-008-9204-0}%
{doi:10.1007/s00453-008-9204-0}.

\bibitem{auguston70}
J.~G. Augustson and J.~Minker.
\newblock {An analysis of some graph theoretical cluster techniques}.
\newblock {\em J. ACM} 17(4):571{--}588, 1970,
  \href{http://dx.doi.org/10.1145/321607.321608}%
{doi:10.1145/321607.321608}.

\bibitem{barabasi99}
A.-L. Barab{\'a}si and R.~Albert.
\newblock {Emergence of scaling in random networks}.
\newblock {\em Science} 286:509{--}512, 1999,
  \href{http://dx.doi.org/10.1126/science.286.5439.509}%
{doi:10.1126/science.286.5439.509}.

\bibitem{batagelj2003}
V.~Batagelj and M.~Zaver{\v{s}}nik.
\newblock {An $O(m)$ algorithm for cores decomposition of networks}, 2003,
  \href{http://arxiv.org/abs/cs/0310049}{arXiv:cs/0310049}.

\bibitem{BerKoMoy-WAO-04}
N.~M. Berry, T.~H. Ko, T.~Moy, J.~Smrcka, J.~Turnley, and B.~Wu.
\newblock {Emergent clique formation in terrorist recruitment}.
\newblock {\em Proc. AAAI-04 Worksh. Agent Organizations}. AAAI Press, 2004,
  \url{http://www.aaai.org/Papers/Workshops/2004/WS-04-02/WS04-02-005.pdf}.

\bibitem{bron73}
C.~Bron and J.~Kerbosch.
\newblock {Algorithm 457: finding all cliques of an undirected graph}.
\newblock {\em Commun. ACM} 16(9):575{--}577, 1973,
  \href{http://dx.doi.org/10.1145/362342.362367}%
{doi:10.1145/362342.362367}.

\bibitem{CaiChaCha-IWPEC-06}
L.~Cai, S.~Chan, and S.~Chan.
\newblock {Random separation: A new method for solving fixed-cardinality
  optimization problems}.
\newblock {\em Proc. 2nd Int. Worksh. Parameterized and Exact Computation
  (IWPEC 2006)}, pp.~239{--}250. Springer-Verlag, LNCS 4169, 2006,
  \href{http://dx.doi.org/10.1007/11847250\_22}%
{doi:10.1007/11847250\_22}.

\bibitem{cazals2008}
F.~Cazals and C.~Karande.
\newblock {A note on the problem of reporting maximal cliques}.
\newblock {\em Theor. Comput. Sci.} 407(1-3):564 -- 568, 2008,
  \href{http://dx.doi.org/10.1016/j.tcs.2008.05.010}%
{doi:10.1016/j.tcs.2008.05.010}.

\bibitem{chiba85}
N.~Chiba and T.~Nishizeki.
\newblock {Arboricity and subgraph listing algorithms}.
\newblock {\em SIAM J. Comput.} 14(1):210{--}223, 1985,
  \href{http://dx.doi.org/10.1137/0214017}%
{doi:10.1137/0214017}.

\bibitem{chrobak91}
M.~Chrobak and D.~Eppstein.
\newblock {Planar orientations with low out-degree and compaction of adjacency
  matrices}.
\newblock {\em Theor. Comput. Sci.} 86(2):243 -- 266, 1991,
  \href{http://dx.doi.org/10.1016/0304-3975(91)90020-3}%
{doi:10.1016/0304-3975(91)90020-3}.

\bibitem{downey95}
R.~G. Downey and M.~R. Fellows.
\newblock {Fixed-parameter tractability and completeness II: On completeness
  for W[1]}.
\newblock {\em Theor. Comput. Sci.} 141(1-2):109 -- 131, 1995,
  \href{http://dx.doi.org/10.1016/0304-3975(94)00097-3}%
{doi:10.1016/0304-3975(94)00097-3}.

\bibitem{DowFel-99}
R.~G. Downey and M.~R. Fellows.
\newblock {\em {Parameterized Complexity}}.
\newblock Springer-Verlag, 1999.

\bibitem{DuWuPei-WMSNA-07}
N.~Du, B.~Wu, X.~Pei, B.~Wang, and L.~Xu.
\newblock {Community detection in large-scale social networks}.
\newblock {\em Proc. 9th WebKDD and 1st SNA-KDD 2007 Workshop on Web Mining and
  Social Network Analysis}, pp.~16{--}25, 2007,
  \href{http://dx.doi.org/10.1145/1348549.1348552}%
{doi:10.1145/1348549.1348552}.

\bibitem{Epp-JGAA-03}
D.~Eppstein.
\newblock {Small maximal independent sets and faster exact graph coloring}.
\newblock {\em J. Graph Algorithms {\&} Applications} 7(2):131{--}140, 2003,
  \href{http://arxiv.org/abs/cs.DS/0011009}{arXiv:cs.DS/0011009}.

\bibitem{Epp-TALG-09}
D.~Eppstein.
\newblock {All maximal independent sets and dynamic dominance for sparse
  graphs}.
\newblock {\em ACM Trans. Algorithms} 5(4):A38, 2009,
  \href{http://dx.doi.org/10.1145/1597036.1597042}%
{doi:10.1145/1597036.1597042},
  \href{http://arxiv.org/abs/cs.DS/0407036}{arXiv:cs.DS/0407036}.

\bibitem{eppstein2009}
D.~Eppstein and E.~S. Spiro.
\newblock {The $h$-index of a graph and its application to dynamic subgraph
  statistics}.
\newblock {\em Proc. 11th Symp. Algorithms and Data Structures (WADS 2009)},
  pp.~278{--}289. Springer-Verlag, LNCS 5664, 2009,
  \href{http://dx.doi.org/10.1007/978-3-642-03367-4\_25}%
{doi:10.1007/978-3-642-03367-4\_25}.

\bibitem{erdos66}
P.~Erd{\H{o}}s and A.~Hajnal.
\newblock {On chromatic number of graphs and set-systems}.
\newblock {\em Acta Mathematica Hungarica} 17(1{--}2):61{--}99, 1966,
  \href{http://dx.doi.org/10.1007/BF02020444}%
{doi:10.1007/BF02020444}.

\bibitem{frank91}
O.~Frank.
\newblock {Statistical analysis of change in networks}.
\newblock {\em Statistica Neerlandica} 45(3):283{--}293, 1991,
  \href{http://dx.doi.org/10.1111/j.1467-9574.1991.tb01310.x}%
{doi:10.1111/j.1467-9574.1991.tb01310.x}.

\bibitem{frank86}
O.~Frank and D.~Strauss.
\newblock {Markov graphs}.
\newblock {\em J. Am. Stat. Assoc.} 81(395):832{--}842, 1986,
  \href{http://dx.doi.org/10.2307/2289017}%
{doi:10.2307/2289017}.

\bibitem{freuder82}
E.~C. Freuder.
\newblock {A sufficient condition for backtrack-free search}.
\newblock {\em J. ACM} 29(1):24{--}32, 1982,
  \href{http://dx.doi.org/10.1145/322290.322292}%
{doi:10.1145/322290.322292}.

\bibitem{gardiner99}
E.~J. Gardiner, P.~Willett, and P.~J. Artymiuk.
\newblock {Graph-theoretic techniques for macromole- cular docking}.
\newblock {\em J. Chem. Inf. Comput. Sci.} 40(2):273{--}279, 2000,
  \href{http://dx.doi.org/10.1021/ci990262o}%
{doi:10.1021/ci990262o}.

\bibitem{gerhards79}
L.~Gerhards and W.~Lindenberg.
\newblock {Clique detection for nondirected graphs: Two new algorithms}.
\newblock {\em Computing} 21(4):295{--}322, 1979,
  \href{http://dx.doi.org/10.1007/BF02248731}%
{doi:10.1007/BF02248731}.

\bibitem{goel2006}
G.~Goel and J.~Gustedt.
\newblock {Bounded arboricity to determine the local structure of sparse
  graphs}.
\newblock {\em WG 2006}, pp.~159{--}167. Springer-Verlag, LNCS 4271, 2006,
  \href{http://dx.doi.org/10.1007/11917496\_15}%
{doi:10.1007/11917496\_15}.

\bibitem{GolVil-WG-08}
P.~A. Golovach and Y.~Villanger.
\newblock {Parameterized complexity for domination problems on degenerate
  graphs}.
\newblock {\em Proc. 34th Int. Worksh. Graph-Theoretic Concepts in Computer
  Science (WG 2008)} 5344:195{--}205, 2008,
  \href{http://dx.doi.org/10.1007/978-3-540-92248-3\_18}%
{doi:10.1007/978-3-540-92248-3\_18}.

\bibitem{grindley93}
H.~M. Grindley, P.~J. Artymiuk, D.~W. Rice, and P.~Willett.
\newblock {Identification of tertiary structure resemblance in proteins using a
  maximal common subgraph isomorphism algorithm}.
\newblock {\em J. Mol. Biol.} 229(3):707 -- 721, 1993,
  \href{http://dx.doi.org/10.1006/jmbi.1993.1074}%
{doi:10.1006/jmbi.1993.1074}.

\bibitem{harary72}
F.~Harary.
\newblock {\em {Graph Theory}}.
\newblock Addison-Wesley, Reading, MA, 1972.

\bibitem{harary57}
F.~Harary and I.~C. Ross.
\newblock {A procedure for clique detection using the group matrix}.
\newblock {\em Sociometry} 20(3):205{--}215, 1957,
  \href{http://dx.doi.org/10.2307/2785673}%
{doi:10.2307/2785673}.

\bibitem{HorSko-PAMI-89}
R.~Horaud and T.~Skordas.
\newblock {Stereo correspondence through feature grouping and maximal cliques}.
\newblock {\em IEEE Trans. Patt. An. Mach. Int.} 11(11):1168{--}1180, 1989,
  \href{http://dx.doi.org/10.1109/34.42855}%
{doi:10.1109/34.42855}.

\bibitem{jensen95}
T.~R. Jensen and B.~Toft.
\newblock {\em {Graph Coloring Problems}}.
\newblock Wiley-Interscience, New York, 1995.

\bibitem{johnson88}
D.~S. Johnson, M.~Yannakakis, and C.~H. Papadimitriou.
\newblock {On generating all maximal in- dependent sets}.
\newblock {\em Inf. Proc. Lett.} 27(3):119 -- 123, 1988,
  \href{http://dx.doi.org/10.1016/0020-0190(88)90065-8}%
{doi:10.1016/0020-0190(88)90065-8}.

\bibitem{johnston75}
H.~C. Johnston.
\newblock {Cliques of a graph{---}variations on the Bron{--}Kerbosch
  algorithm}.
\newblock {\em Int. J. Parallel Programming} 5(3):209{--}238, 1976,
  \href{http://dx.doi.org/10.1007/BF00991836}%
{doi:10.1007/BF00991836}.

\bibitem{kirousis96}
L.~Kirousis and D.~Thilikos.
\newblock {The linkage of a graph}.
\newblock {\em SIAM J. Comput.} 25(3):626{--}647, 1996,
  \href{http://dx.doi.org/10.1137/S0097539793255709}%
{doi:10.1137/S0097539793255709}.

\bibitem{KloCai-ICS-00}
T.~Kloks and L.~Cai.
\newblock {Parameterized tractability of some (efficient) $Y$-domination
  variants for planar graphs and $t$-degenerate graphs}.
\newblock {\em Proc. International Computer Symposium}, 2000,
  \url{http://hdl.handle.net/2377/2482}.

\bibitem{koch2001}
I.~Koch.
\newblock {Enumerating all connected maximal common subgraphs in two graphs}.
\newblock {\em Theor. Comput. Sci.} 250(1-2):1 -- 30, 2001,
  \href{http://dx.doi.org/10.1016/S0304-3975(00)00286-3}%
{doi:10.1016/S0304-3975(00)00286-3}.

\bibitem{KocLenWan-JCB-96}
I.~Koch, T.~Lengauer, and E.~Wanke.
\newblock {An algorithm for finding maximal common subtopologies in a set of
  protein structures}.
\newblock {\em J. Comput. Biol.} 3(2):289{--}306, 1996,
  \href{http://dx.doi.org/10.1089/cmb.1996.3.289}%
{doi:10.1089/cmb.1996.3.289}.

\bibitem{lawler80}
E.~L. Lawler, J.~K. Lenstra, and A.~H.~G. Rinnooy~Kan.
\newblock {Generating all maximal independent sets: NP-hardness and
  polynomial-time algorithms}.
\newblock {\em SIAM J. Comput.} 9(3):558{--}565, 1980,
  \href{http://dx.doi.org/10.1137/0209042}%
{doi:10.1137/0209042}.

\bibitem{LicWhi-CJM-70}
D.~R. Lick and A.~T. White.
\newblock {$k$-degenerate graphs}.
\newblock {\em Canad. J. Math.} 22:1082{--}1096, 1970,
  \url{http://www.smc.math.ca/cjm/v22/p1082}.

\bibitem{loukakis81}
E.~Loukakis and C.~Tsouros.
\newblock {A depth first search algorithm to generate the family of maximal
  independent sets of a graph lexicographically}.
\newblock {\em Computing} 27(4):349{--}366, 1981,
  \href{http://dx.doi.org/10.1007/BF02277184}%
{doi:10.1007/BF02277184}.

\bibitem{luce49}
R.~D. Luce and A.~D. Perry.
\newblock {A method of matrix analysis of group structure}.
\newblock {\em Psychometrika} 14(2):95{--}116, 1949,
  \href{http://dx.doi.org/10.1007/BF02289146}%
{doi:10.1007/BF02289146}.

\bibitem{makino2004}
K.~Makino and T.~Uno.
\newblock {New algorithms for enumerating all maximal cliques}.
\newblock {\em Proc. 9th Scand. Worksh. Algorithm Theory}, pp.~260{--}272.
  Springer-Verlag, LNCS 3111, 2004.

\bibitem{moon-moser65}
J.~W. Moon and L.~Moser.
\newblock {On cliques in graphs}.
\newblock {\em Israel J. Math.} 3(1):23{--}28, 1965,
  \href{http://dx.doi.org/10.1007/BF02760024}%
{doi:10.1007/BF02760024}.

\bibitem{mulligan72}
G.~D. Mulligan and D.~G. Corneil.
\newblock {Corrections to Bierstone's algorithm for generating cliques}.
\newblock {\em J. ACM} 19(2):244{--}247, 1972,
  \href{http://dx.doi.org/10.1145/321694.321698}%
{doi:10.1145/321694.321698}.

\bibitem{robins2007}
G.~Robins and M.~Morris.
\newblock {Advances in exponential random graph ($p^{*}$) models}.
\newblock {\em Social Networks} 29(2):169{--}172, 2007,
  \href{http://dx.doi.org/10.1016/j.socnet.2006.08.004}%
{doi:10.1016/j.socnet.2006.08.004}.

\bibitem{samudrala98}
R.~Samudrala and J.~Moult.
\newblock {A graph-theoretic algorithm for comparative modeling of protein
  structure}.
\newblock {\em J. Mol. Biol.} 279(1):287 -- 302, 1998,
  \href{http://dx.doi.org/10.1006/jmbi.1998.1689}%
{doi:10.1006/jmbi.1998.1689}.

\bibitem{BioGRID2006}
C.~Stark, B.-J. Breitkreutz, T.~Reguly, L.~Boucher, A.~Breitkreutz, and
  M.~Tyers.
\newblock {BioGRID: a general repository for interaction datasets}.
\newblock {\em Nucleic Acids Res.} 34:D535{--}D539, 2006,
  \href{http://dx.doi.org/10.1093/nar/gkj109}%
{doi:10.1093/nar/gkj109}.

\bibitem{tomita2006}
E.~Tomita, A.~Tanaka, and H.~Takahashi.
\newblock {The worst-case time complexity for generating all maximal cliques
  and computational experiments}.
\newblock {\em Theor. Comput. Sci.} 363(1):28{--}42, 2006,
  \href{http://dx.doi.org/10.1016/j.tcs.2006.06.015}%
{doi:10.1016/j.tcs.2006.06.015}.

\bibitem{tsukiyama77}
S.~Tsukiyama, M.~Ide, H.~Ariyoshi, and I.~Shirakawa.
\newblock {A new algorithm for generating all the maximal independent sets}.
\newblock {\em SIAM J. Comput.} 6(3):505{--}517, 1977,
  \href{http://dx.doi.org/10.1137/0206036}%
{doi:10.1137/0206036}.

\bibitem{wasserman96}
S.~Wasserman and P.~Pattison.
\newblock {Logit models and logistic regressions for social networks: I. An
  introduction to Markov graphs and $p^{*}$}.
\newblock {\em Psychometrika} 61(3):401{--}425, 1996,
  \href{http://dx.doi.org/10.1007/BF02294547}%
{doi:10.1007/BF02294547}.

\bibitem{Woo-GC-07}
D.~R. Wood.
\newblock {On the maximum number of cliques in a graph}.
\newblock {\em Graphs and Combinatorics} 23(3):337{--}352, 2007,
  \href{http://dx.doi.org/10.1007/s00373-007-0738-8}%
{doi:10.1007/s00373-007-0738-8}.

\bibitem{ZakParOgi-KDD-97}
M.~J. Zaki, S.~Parthasarathy, M.~Ogihara, and W.~Li.
\newblock {New algorithms for fast discovery of association rules}.
\newblock {\em Proc. 3rd Int. Conf. Knowledge Discovery and Data Mining},
  pp.~283{--}286. AAAI Press, 1997,
  \url{http://www.aaai.org/Papers/KDD/1997/KDD97-060.pdf}.

\bibitem{Zom-SoCG-10}
A.~Zomorodian.
\newblock {The tidy set: a minimal simplicial set for computing homology of
  clique complexes}.
\newblock {\em Proc. 26th ACM Symp. Computational Geometry}, pp.~257{--}266,
  2010, \url{http://www.cs.dartmouth.edu/~afra/papers/socg10/tidy-socg.pdf}.

\end{thebibliography}
}

\clearpage
\begin{appendix}

\section{Appendix}
\label{sec:appendix}
\subsection{The Degeneracy of Protein--Protein Interaction Networks}
The Biological General Repository for Interaction Datasets (BioGRID)~\cite{BioGRID2006}~\footnote{\url{http://thebiogrid.org/}} is a 
curated database containing the interactions between proteins that have 
been published in the literature, and those that have been discovered 
through high-throughput screening methods. Using this interaction data, 
we can form a graph by creating a vertex for each protein and an edge 
between two proteins that interact. Such a graph is called a protein--protein 
interaction (PPI) network. 

We computed the degeneracy of seven PPI networks in version 3.0.65 
of the BioGRID database. We chose to omit other networks in the BioGRID 
database because they contain only a handful of known proteins and 
interactions. The results are summarized in Table~\ref{table:ppi}. Observe that,
for each PPI network, the degeneracy is significantly lower than both the 
number of vertices in the graph and the maximum degree of the graph.
We therefore conclude that the degeneracy of these PPI networks is low and that fixed-parameter algorithms parametrized by degeneracy can be expected to perform well on these graphs.
\begin{table}
\begin{center}
\caption{Graph statistics for seven PPI networks from 
version 3.0.65 of the BioGRID database. 
For each network we list the number of vertices ($n$), the number of edges 
($m$), the maximum degree ($\Delta$), and the degeneracy ($d$).}
\label{table:ppi}
\begin{tabular}{l|c|c|c|c}
Organism (PPI network)      & $n$ & $m$ & $\Delta$ & $d$\\ \hline
Mus musculus                & 1455 & 1636   & 111  & 6 \\
Caenorhabditis elegans      & 3518 & 6531   & 523  & 10 \\
Arabisopsis thaliana        & 1745 & 3098   & 71   & 12 \\
Drosophila melanogaster     & 7282 & 24894  & 176  & 12 \\
Homo Sapiens                & 9527 & 31182  & 308  & 12 \\
Schizosaccharomyces pombe   & 2031 & 12637  & 439  & 34 \\
Saccharomyces cerevisae     & 6008 & 156945 & 2557 & 64 
\end{tabular}
\end{center}
\end{table}

\end{appendix}

\end{document}